\documentclass[11pt]{article}

\usepackage[T1]{fontenc}
\usepackage{lmodern}
\usepackage[letterpaper,margin=1.05in]{geometry}
\usepackage{microtype}
\usepackage{amsmath,amssymb,amsthm,mathtools}
\usepackage{booktabs}
\usepackage{array}
\usepackage{enumitem}
\usepackage{xcolor}
\usepackage{natbib}
\usepackage[hidelinks,pdfusetitle]{hyperref}
\usepackage[nameinlink,noabbrev]{cleveref}

\setlist{itemsep=0.25em,topsep=0.4em}
\allowdisplaybreaks

\newtheorem{theorem}{Theorem}[section]
\newtheorem{proposition}[theorem]{Proposition}
\newtheorem{lemma}[theorem]{Lemma}
\newtheorem{corollary}[theorem]{Corollary}
\theoremstyle{definition}

\newtheorem{example}[theorem]{Example}
\theoremstyle{remark}
\newtheorem{remark}[theorem]{Remark}

\newcommand{\R}{\mathbb{R}}
\newcommand{\eps}{\epsilon}
\newcommand{\pos}[1]{\left(#1\right)^{+}}
\newcommand{\negp}[1]{\left(#1\right)^{-}}
\newcommand{\bid}{\mathrm{bid}}
\newcommand{\ask}{\mathrm{ask}}
\newcommand{\cW}{\mathcal{W}}
\newcommand{\cQ}{\mathcal{Q}}
\newcommand{\cD}{\mathsf{D}}

\newcommand{\CVB}{\mathrm{CVB}}

\newcommand{\supp}{\operatorname{supp}}

\newcommand{\ind}{\mathbf{1}}

\newcolumntype{L}[1]{>{\raggedright\arraybackslash}p{#1}}

\crefname{theorem}{theorem}{theorems}
\crefname{proposition}{proposition}{propositions}
\crefname{lemma}{lemma}{lemmas}
\crefname{corollary}{corollary}{corollaries}
\crefname{definition}{definition}{definitions}
\crefname{example}{example}{examples}
\crefname{remark}{remark}{remarks}

\title{\textbf{Multi-maturity consistency of option prices under bounded
bid--ask spreads: a minimal obstruction and an exact two-date basket
operator}}
\author{Minhyeok Lee
\\[0.25em]
{\small Independent Researcher, Republic of Korea}
\\[-0.05em]
{\small\texttt{borrownotime@gmail.com}}
}
\date{30 July 2026}

\hypersetup{
  pdfauthor={Minhyeok Lee},
  pdfsubject={Multi-maturity option-price consistency under bounded bid--ask spreads},
  pdfkeywords={option-price consistency, model-independent arbitrage,
    bid--ask spreads, transaction costs, robust superhedging, convex order}
}

\begin{document}

\maketitle

\begin{abstract}
Gerhold and G\"ul\"um derived necessary calendar--vertical--basket
conditions for finite call bid--ask quotes when the cash-settlement reference
price lies inside a dynamically traded stock spread of bounded absolute
width.  We first distinguish the printed bid of a calendar--vertical basket
from the executable bid dictated by the contract and the self-financing
convention.  After making this correction and adjoining the initial-spread
and complete one-maturity conditions, we show that the resulting system is
still insufficient.  Thus, under the corrected executable reading, the
sufficiency direction of Conjecture~5.4 of Gerhold and G\"ul\"um has a
negative answer even after the natural base conditions are imposed; its
separate weak-arbitrage clause is not addressed.  For every positive spread
bound, an explicit panel with two quoted dates and one actual call at each
date satisfies all corrected conditions, strictly whenever a strict face
applies, but admits a pathwise stock-flip arbitrage.  Measured by quoted dates
and actual calls, this obstruction is minimal within the finite-call
framework of Gerhold and G\"ul\"um, and the counterexamples contain a
full-dimensional open quote box.  A separate extremal-envelope argument
produces the same separation gap.

We then eliminate the arbitrary adapted stock holding in the complete
two-date pathwise problem.  The result is a closed-form one-step operator,
expressible either as a finite concave maximization or as a convex-envelope
infimum over a \(2\eps\)-neighborhood.  As the finite call positions vary,
the operator characterizes the full two-date executable
model-independent-arbitrage basket cone.  General robust superhedging
duality and backward principles are prior work; the contribution here is the
explicit calculation for this bounded-spread reference/shadow geometry.
\end{abstract}

\noindent\textbf{Keywords.}
consistency of option prices; model-independent arbitrage; option bid--ask
spreads; transaction costs; robust superhedging; convex order

\medskip
\noindent\textbf{2020 Mathematics Subject Classification.}
Primary 91G20; secondary 60E15.

\section{Introduction}
\label{sec:introduction}

Finite option quotes encode a geometric no-arbitrage problem.  In a
frictionless one-period market, monotonicity, convexity, and the slope bounds
of the call-price curve lead to explicit consistency conditions
\citep{CarrMadan2005,DavisHobson2007,Cousot2007}.  With several maturities,
convex order supplies the intertemporal restriction.  Bid--ask spreads make
both parts subtler: long and short static positions execute at different
prices, and dynamic stock trades cross a spread whose orientation is observed
only when the trade is made.

\citet{GerholdGulum2020} study a particularly sharp version of this problem.
Calls are cash settled against a nontraded reference price \(S^C\), while a
shadow martingale \(S^*\) must remain inside the bid--ask spread of the
dynamically traded stock.  The discounted spread has a deterministic absolute
width bound \(\eps\).  Their one-maturity conditions are complete.  For several
maturities they introduce \emph{calendar--vertical baskets} (CVBs) and derive
four families of necessary inequalities.  Conjecture~5.4 of
\citet{GerholdGulum2020} has two parts: it asks whether those four families
are also sufficient for consistency, and it predicts weak arbitrage when the
first three hold but the strict fourth condition fails.

A sign inconsistency in the cited CVB bid must first be resolved.  The
contract contains a cash term \(-2\eps\ind_{\{\sigma_1=-1\}}\), while the
displayed bid in
\citet[Definition~5.1 and equation~(22)]{GerholdGulum2020} adds
\(2\eps\) on that face.  The contract, the paper's self-financing convention,
and the liquidation calculation in its Lemma~5.2 instead give an executable
bid with a minus sign.  The distinction is substantive: the printed bid can
violate a claimed necessary inequality in a constant, perfectly consistent
model.  We therefore study the corrected executable question: after replacing
the sign and adjoining the initial and complete one-maturity requirements,
are the resulting conditions sufficient?

The answer is no at the smallest possible multi-date quote complexity.  The
economic obstruction is a \emph{stock flip}.  Shorting an earlier call and
buying a fractional later call allows a trader to hold no stock after a low
first-date state and a short stock position after a high state.  A single CVB
comparison does not encode this coupled switch.

\begin{theorem}[Main result]
\label{thm:main}
For every \(\eps>0\), there is a two-date panel with one actual call quote at
each date such that:
\begin{enumerate}[label=\textup{(\roman*)}]
    \item the initial spread and the complete one-maturity consistency
    conditions hold;
    \item all corrected executable CVB conditions hold, strictly whenever
    the source requires a strict face;
    \item the strategy consisting of \(2/3\) of a share, one short
    date-one call, and \(1/3\) of a long date-two call has setup cost
    \[
        -\frac{47}{150}\eps
    \]
    and nonnegative liquidation on every admissible path;
    \item this counterexample persists on a six-dimensional open box of
    quote panels.
\end{enumerate}
Two quoted dates and two actual calls, necessarily split one per date, are
minimal for such an augmented intertemporal counterexample.

Moreover, for every payoff \(h_2\) generated by a finite terminal call
portfolio, every initial stock holding \(\theta\), and \(x\ge\eps\), the
two-date continuation
operator is the following.  Extend \(h_2\) by \(+\infty\) on
\((-\infty,\eps)\) before taking the biconjugate \(h_2^{**}\); then
\[
\begin{aligned}
(\cW_{\eps,\theta}h_2)(x)
={}&\theta x-\eps|\theta|\\
&+\inf_{\substack{z\ge\eps\\ |z-x|\le2\eps}}
\left\{h_2^{**}(z)+\pos{\theta(z-x)}\right\}.
\end{aligned}
\]
Together with the date-one static payoff, this operator characterizes the
full two-date finite-call model-independent-arbitrage basket cone.
\end{theorem}

\noindent
The finite-panel assertions and minimality are proved in
\Cref{thm:minimal-panel}, the open-box statement in \Cref{cor:open-box}, and
the operator assertions in \Cref{thm:operator} and
\Cref{prop:convex-envelope}.

The counterexample is not isolated.  The six quote coordinates may be
perturbed independently while preserving every tested condition and the
negative setup cost.  Nor is its certificate tied to the trading strategy:
the extremal call-function envelopes from
\citet{GerholdGulum2019} force a convex-order gap of
\(47\eps/50\); after multiplication by the terminal call weight \(1/3\),
this equals the initial gain \(47\eps/150\).

\paragraph{Relation to prior work.}
Consistent price systems and
fundamental theorems under transaction costs have a long history
\citep{JouiniKallal1995,Schachermayer2004,PennanenPenner2010,
BouchardNutz2016,BayraktarZhang2016}.  Probability-free and robust
superhedging recursions under frictions appear in
\citet{DolinskySoner2014} and \citet{Burzoni2016}.
\citet{CheriditoKupperTangpi2017} allow state-dependent convex dynamic costs
and finite static portfolios in arbitrary derivatives, a framework that
contains the present execution architecture after a direct encoding.
Specializing that architecture to the bounded reference/shadow spread yields
the penalty \(-\eps(|q-\theta|+|q|)\), the scalar convex-envelope formula,
and a finite evaluation rule.
An explicit restricted-support convex-envelope recursion also appears in
Chapter~5 of \citet{Bauer2024}.

Adjacent finite-quote work uses different market assumptions.
\citet{CohenReisingerWang2020} treat static repair with a frictionless
underlying.  Computational martingale optimal transport fixes marginal laws
and relaxes martingale constraints for numerical purposes
\citep{GuoObloj2019}.  The current bid--ask martingale optimal transport
duality of \citet{LiangEtAl2026} allows vanilla bid--ask execution but keeps
dynamic trading in the underlying frictionless.  None of these formulations
contains simultaneously the absolute stock-spread bound, the distinct
reference and shadow prices, and the finite execution structure studied
here.

\paragraph{Organization.}
\Cref{sec:market} gives a self-contained discounted market and separates the
printed and executable CVB systems.  \Cref{sec:minimal} proves the stock-flip
bridge, the counterexample family for every \(\eps>0\), and its open
neighborhood.
\Cref{sec:envelope} gives the independent convex-order certificate.
\Cref{sec:operator} derives the two-date operator, including the
compact minimax reduction, conjugate representation, attainment, and Borel
selector.  \Cref{sec:limits} shows that adding every pairwise stock-flip
bridge still does not close the full basket cone, and records the
scope of the results.

\paragraph{AI-assisted research.}
This paper was developed with substantive assistance from generative AI.
The configurations used and their roles are described in
\Cref{sec:aiuse}.

\section{Market, execution, and the two CVB systems}
\label{sec:market}

We work in discounted units on dates \(0,1,2\); the bank account is one.
Fix \(\eps>0\).  At date \(t=1,2\), an admissible stock state is
\[
    (\underline s_t,x_t,\overline s_t)
\]
with
\begin{equation}
0<\underline s_t\le x_t\le\overline s_t,\qquad
\overline s_t-\underline s_t\le\eps,\qquad x_t\ge\eps.
\label{eq:admissible-state}
\end{equation}
Here \(x_t\) is the cash-settlement reference price.  Calls with strike
\(k_{t,i}>\eps\) pay \(\pos{x_t-k_{t,i}}\) and are quoted at
\[
    [\,\underline r_{t,i},\overline r_{t,i}\,],
    \qquad 1\le i\le N_t.
\]
Assume \(0<S_0^{\bid}\le S_0^{\ask}\).  The date-zero stock quote is
\([S_0^{\bid},S_0^{\ask}]\).

\subsection{Consistency and one-maturity witnesses}

We recall only the model objects needed below.  For a probability measure
\(\mu\) with finite first moment, write
\[
    R_\mu(k)=\int_{\R}(y-k)^+\,\mu(dy)
\]
for its call function, and let \(W^\infty\) denote the infinity-Wasserstein
distance.  A family \((\nu_t)\) is increasing in convex order, written
\(\nu_1\le_c\nu_2\), if both measures have the same mean and
\(\int f\,d\nu_1\le\int f\,d\nu_2\) for every convex \(f\) for which the
integrals exist.

The finite-support reduction in \citet[Lemma~2.7]{GerholdGulum2020} says that
a panel is \(\eps\)-consistent precisely when the initial stock width is at
most \(\eps\) and there are finitely supported reference laws \(\mu_t\) and
shadow laws \(\nu_t\) such that
\begin{equation}
\begin{gathered}
    \supp\mu_t\subseteq[\eps,\infty),\qquad
    R_{\mu_t}(k_{t,i})
       \in[\underline r_{t,i},\overline r_{t,i}],\\
    W^\infty(\mu_t,\nu_t)\le\eps,\qquad
    \nu_1\le_c\nu_2,\qquad
    E\nu_2\in[S_0^{\bid},S_0^{\ask}].
\end{gathered}
\label{eq:measure-consistency}
\end{equation}
For a single quoted maturity, the corresponding condition is obtained by
dropping convex order and requiring the mean of its shadow law to lie in the
initial interval.  We call the initial-width requirement together with this
complete one-maturity requirement at every quoted date the
\emph{base system}.

\subsection{Executable semistatic portfolios}

Let \(c\in\R\) be the initial bank holding, \(\theta\in\R\) the number of
shares held from date zero to date one, and
\(\eta_{t,i}\in\R\) the static call positions.  Their date-zero executable
cost is
\begin{equation}
\begin{aligned}
\operatorname{Cost}(c,\theta,\eta)
={}&c+\theta^+S_0^{\ask}-\theta^-S_0^{\bid}\\
&+\sum_{t=1}^2\sum_{i=1}^{N_t}
\left(\eta_{t,i}^+\overline r_{t,i}
      -\eta_{t,i}^-\underline r_{t,i}\right).
\end{aligned}
\label{eq:setup-cost}
\end{equation}
After observing the complete date-one state, the investor may rebalance to a
Borel-measurable holding \(q\).  The rebalance is made at the date-one ask
when \(q-\theta>0\) and at the bid when \(q-\theta<0\).  The remaining stock
is liquidated at date two.  A \emph{model-independent arbitrage} is a
self-financing semistatic portfolio with negative setup cost and
nonnegative final liquidation for every pair of states satisfying
\eqref{eq:admissible-state}.  This is stronger than a weak or
model-dependent arbitrage because the same strategy must work on every
admissible path, independently of any probability model.

\subsection{Printed and executable CVB bids}

The sign issue can be seen directly from the contract.  In the notation of
\citet[Definition~5.1]{GerholdGulum2020}, a CVB of maturity \(u<T\) has cash
settlement
\begin{equation}
\CVB_u
=C_1(k_{1,j_1})
+\sum_{v=2}^{u}
\bigl(C_v(k_{v,j_v})-C_v(k_{v,i_v})\bigr)
-2\eps\ind_{\{\sigma_1=-1\}}.
\label{eq:cvb-contract}
\end{equation}
The indices and signs obey the path-matching restrictions in that
definition; their details are immaterial for the accounting identity.

\begin{proposition}[Printed versus executable bid]
\label{prop:sign}
The amount received by selling \(\CVB_u\), before following the dynamic
liquidation rule, is
\begin{equation}
\underline r_u^{\CVB,\mathrm{exec}}
=\underline r_{1,j_1}
+\sum_{v=2}^{u}
\bigl(\underline r_{v,j_v}-\overline r_{v,i_v}\bigr)
-2\eps\ind_{\{\sigma_1=-1\}}.
\label{eq:executable-cvb-bid}
\end{equation}
The \(+2\eps\) adjustment in the printed CVB bid is inconsistent with
\eqref{eq:cvb-contract}, the executable option accounting, and the
self-financing identity.
\end{proposition}

\begin{proof}
Selling \eqref{eq:cvb-contract} reverses every static position.  Its
date-zero cost, using option bids for shorts and asks for longs, is
\[
-\underline r_{1,j_1}
+\sum_{v=2}^{u}
\bigl(\overline r_{v,i_v}-\underline r_{v,j_v}\bigr)
+2\eps\ind_{\{\sigma_1=-1\}}.
\]
This is the negative of \eqref{eq:executable-cvb-bid}.  It is also the
initial-value identity asserted in the source's Lemma~5.2.  The dynamic
stock calculation in that proof is consistent with this expression.
\end{proof}

The distinction changes necessity, not merely notation.

\begin{example}[A consistent constant model violating the printed necessary
condition]
\label{ex:printed-sign}
Take \(\eps=1\), a constant reference and shadow price equal to \(5\), zero
stock spread, and a strike-\(2\) call worth \(3\) at both dates.  This is a
consistent model.  For the face \(\sigma_1=-1\), however, the printed formula
assigns the earlier CVB the bid \(3+2=5\) at shifted strike \(2+1=3\).
The later shifted ask at the same point is \(3\), so the printed monotonicity
condition would require \(5\le3\).  The executable bid is instead
\(3-2=1\), as required by the contract.
\end{example}

\begin{remark}[Source versions]
\label{rem:source-versions}
The same sign pattern appears in
\citet[Definition~3.10, equation~(3.20), and
Proposition~3.11]{GulumThesis2016}, in arXiv:1608.05585v1 and v2, and in the
published article of \citet{GerholdGulum2020}.  In each case the contract and
self-financing calculation use \(-2\eps\), while the displayed bid uses
\(+2\eps\).  This is an algebraic comparison of the formulas, not an
attribution of authorial intent.
\end{remark}

Accordingly, the result below has two separate implications.  The literally
printed system is not necessary.  After replacing its bid by
\eqref{eq:executable-cvb-bid}, necessity survives, but even the
base-augmented executable system is not sufficient.

\subsection{The complete executable two-date CVB geometry}
\label{subsec:t2-cvb}

For a two-date panel, introduce the source's auxiliary quote at each date:
\begin{equation}
k_{t,0}=\eps,\qquad
\underline r_{t,0}=S_0^{\bid}-2\eps,\qquad
\overline r_{t,0}=S_0^{\ask}.
\label{eq:auxiliary-quote}
\end{equation}
Every CVB then has maturity one and is indexed by
\(j\in\{0,\ldots,N_1\}\) and \(\sigma\in\{-1,1\}\).  Define its shifted
strike and executable bid by
\begin{equation}
z_{j,\sigma}=k_{1,j}-\eps\sigma,\qquad
b_{j,\sigma}=\underline r_{1,j}
-2\eps\ind_{\{\sigma=-1\}}.
\label{eq:t2-cvb-point}
\end{equation}
The shifted terminal ask nodes are
\begin{equation}
q_i=k_{2,i}+\eps,\qquad a_i=\overline r_{2,i},
\qquad 0\le i\le N_2.
\label{eq:terminal-ask-nodes}
\end{equation}

In this notation, the four necessary conditions of
\citet[Theorem~5.3]{GerholdGulum2020}, with the executable correction, are
equivalent to the following call-function geometry for every CVB point
\((z,b)=(z_{j,\sigma},b_{j,\sigma})\):
\begin{align}
b&\le
\frac{(q_\ell-z)a_i+(z-q_i)a_\ell}{q_\ell-q_i},
&&q_i<z<q_\ell,
\label{eq:cvb-chord}\\
b&\le a_\ell+q_\ell-z,
&&z<q_\ell,
\label{eq:cvb-slope}\\
b&\le a_i,
&&z\ge q_i,
\label{eq:cvb-monotone}\\
b=a_i&\ \Longrightarrow\ a_i=0,
&&z>q_i.
\label{eq:cvb-strict}
\end{align}
The first inequality is convexity, the second is the lower slope bound
\(-1\), the third is monotonicity, and the fourth is the source's strict
face.  We call the base system together with
\eqref{eq:cvb-chord}--\eqref{eq:cvb-strict} for every admissible index the
\emph{base-augmented executable CVB system}.  This is the precise system
tested in the sequel.

\begin{lemma}[Call-geometry certificate]
\label{lem:call-geometry}
Suppose there is a nonnegative, nonincreasing, convex function \(U\) on
\(\R\) whose slopes are bounded below by \(-1\), such that
\[
    b_{j,\sigma}<U(z_{j,\sigma})
    \quad\text{for every CVB point}
\]
and
\[
    U(q_i)\le a_i
    \quad\text{for every terminal ask node}.
\]
Then \eqref{eq:cvb-chord}--\eqref{eq:cvb-strict} hold, with strictness on
every applicable face for which \(a_i>0\).
\end{lemma}

\begin{proof}
Convexity places \(U(z)\) below the chord between any two nodes surrounding
\(z\); the terminal asks lie above those node values and the CVB bid lies
below \(U(z)\), proving \eqref{eq:cvb-chord}.  The slope lower bound gives
\(U(z)\le U(q_\ell)+q_\ell-z\) when \(z<q_\ell\), which proves
\eqref{eq:cvb-slope}.  Monotonicity gives
\(U(z)\le U(q_i)\) when \(z\ge q_i\), proving
\eqref{eq:cvb-monotone}.  If \(z>q_i\), the assumed strict inequality at the
CVB point rules out equality with a positive terminal ask, which gives
\eqref{eq:cvb-strict}.
\end{proof}

\section{The minimal stock-flip obstruction}
\label{sec:minimal}

The missing inequality is most transparent before any numerical panel is
introduced.

\begin{proposition}[Stock-flip bridge]
\label{prop:sfb}
Let \(s<t\), and suppose an earlier call at strike \(a>\eps\) and a later
call at strike \(K>\eps\) are traded.  If
\[
    a<K+2\eps,
\]
define
\begin{equation}
\gamma=\frac{a-\eps}{K+\eps},
\qquad
\theta=1-\gamma
=\frac{K-a+2\eps}{K+\eps}.
\label{eq:sfb-coefficients}
\end{equation}
Every panel without model-independent arbitrage satisfies
\begin{equation}
\boxed{\;
\theta S_0^{\ask}
-\underline r_s(a)
+\gamma\overline r_t(K)\ge0.
\;}
\label{eq:sfb}
\end{equation}
Equivalently,
\begin{equation}
\boxed{\;
(K-a+2\eps)S_0^{\ask}
-(K+\eps)\underline r_s(a)
+(a-\eps)\overline r_t(K)\ge0.
\;}
\label{eq:sfb-homogeneous}
\end{equation}
\end{proposition}

\begin{proof}
At date zero, buy \(\theta\) shares, short one date-\(s\) call at \(a\),
and buy \(\gamma\) date-\(t\) calls at \(K\).  Hold no initial cash.  At
date \(s\), after observing \(x_s\), act as follows:
\begin{itemize}
    \item if \(x_s<a\), sell the \(\theta\) shares and hold \(q=0\);
    \item if \(x_s\ge a\), sell one share in total and hold
    \(q=-\gamma\).
\end{itemize}
The rule is Borel and adapted.  In the low state the earlier call expires
worthless and the sale proceeds are at least
\(\theta(x_s-\eps)\ge0\).  The later long call is nonnegative.

In the high state, the proceeds from selling one share, net of the short-call
payoff, are at least
\[
    (x_s-\eps)-(x_s-a)=a-\eps.
\]
At date \(t\), covering the short stock and receiving the long-call payoff
is worth at least
\[
    \gamma\bigl[\pos{x_t-K}-(x_t+\eps)\bigr]
    \ge-\gamma(K+\eps)=-(a-\eps).
\]
Thus final liquidation is nonnegative on every admissible path.  Its setup
cost is the left side of \eqref{eq:sfb}, which cannot be negative in the
absence of model-independent arbitrage.
\end{proof}

The trade flips from zero stock to a fractional short stock position
according to the earlier exercise state.  It uses one actual call at each
date and no probabilistic support restriction.

\subsection{A panel for every positive spread bound}

\begin{theorem}[Minimal counterexample for every \(\eps>0\)]
\label{thm:minimal-panel}
For every \(\eps>0\), consider the stock quote
\begin{equation}
[S_0^{\bid},S_0^{\ask}]
=\left[\frac{599}{100}\eps,\frac{601}{100}\eps\right]
\label{eq:minimal-stock-quote}
\end{equation}
and the two actual call quotes
\begin{equation}
\renewcommand{\arraystretch}{1.35}
\begin{array}{c@{\qquad}c@{\qquad}c@{\qquad}c}
\toprule
t & k & \underline r & \overline r\\
\midrule
1 & 3\eps & \frac{499}{100}\eps & \frac{501}{100}\eps\\[1mm]
2 & 5\eps & \frac{199}{100}\eps & \frac{201}{100}\eps\\
\bottomrule
\end{array}
\label{eq:minimal-panel}
\end{equation}
The panel satisfies the base-augmented executable CVB system, with strict
slack on every applicable strict face, but admits a model-independent
arbitrage of setup cost
\[
    -\frac{47}{150}\eps.
\]
Within the finite-call framework of Gerhold and G\"ul\"um, two quoted dates
and two actual call quotes are minimal for an augmented intertemporal
counterexample.
\end{theorem}

\begin{proof}
We prove the base conditions, exhaust the CVBs, construct the arbitrage, and
then prove minimality.

\paragraph{One-maturity consistency.}
At date one, set
\[
\mu_1=\frac34\delta_{\eps}+\frac14\delta_{23\eps},
\qquad
\nu_1=\frac34\delta_{\eps/2}+\frac14\delta_{45\eps/2}.
\]
Then
\[
E\mu_1=\frac{13}{2}\eps,\qquad
R_{\mu_1}(3\eps)=5\eps,\qquad
E\nu_1=6\eps,\qquad
W^\infty(\mu_1,\nu_1)=\frac{\eps}{2}.
\]
At date two, set
\[
    \mu_2=\delta_{7\eps},\qquad \nu_2=\delta_{6\eps}.
\]
Then
\[
R_{\mu_2}(5\eps)=2\eps,\qquad
E\nu_2=6\eps,\qquad
W^\infty(\mu_2,\nu_2)=\eps.
\]
Both reference laws are supported on \([\eps,\infty)\), their call values
lie strictly inside the quote intervals, and both shadow means equal
\(6\eps\), which lies strictly inside \eqref{eq:minimal-stock-quote}.
Thus each maturity is separately consistent with the same initial interval.
The initial stock width is \(\eps/50<\eps\).

\paragraph{CVB conditions.}
The auxiliary data \eqref{eq:auxiliary-quote} are
\[
k_{t,0}=\eps,\qquad
\underline r_{t,0}=\frac{399}{100}\eps,\qquad
\overline r_{t,0}=\frac{601}{100}\eps.
\]
There are exactly four date-one CVBs:
\begin{equation}
\renewcommand{\arraystretch}{1.50}
\begin{array}{c@{\quad}c@{\quad}c@{\quad}c}
\toprule
j&\sigma&z_{j,\sigma}&b_{j,\sigma}\\
\midrule
0&-1&2\eps&\frac{199}{100}\eps\\
0& 1&0&\frac{399}{100}\eps\\
1&-1&4\eps&\frac{299}{100}\eps\\
1& 1&2\eps&\frac{499}{100}\eps\\
\bottomrule
\end{array}
\label{eq:four-cvbs}
\end{equation}
The shifted terminal ask nodes are
\[
\left(2\eps,\frac{601}{100}\eps\right),
\qquad
\left(6\eps,\frac{201}{100}\eps\right).
\]
Now take
\[
    U(x)=\pos{8\eps-x}.
\]
This is a nonnegative, nonincreasing convex call function with slopes in
\([-1,0]\).  It lies strictly below both terminal asks and strictly above
all four CVB points.  \Cref{lem:call-geometry} proves every executable CVB
condition.  The corresponding finite slacks are
\begin{equation}
\renewcommand{\arraystretch}{1.35}
\begin{array}{c@{\qquad}c@{\qquad}l}
\toprule
\text{condition}&\text{number of instances}&\text{positive slack multiset}\\
\midrule
\eqref{eq:cvb-chord}&1&
\left\{\frac{51}{50}\eps\right\}\\[1mm]
\eqref{eq:cvb-slope}&5&
\left\{\frac{51}{50}\eps\right\}^{\times2},
\left\{\frac{201}{50}\eps\right\}^{\times3}\\[1mm]
\eqref{eq:cvb-monotone}&3&
\left\{\frac{51}{50},\frac{151}{50},\frac{201}{50}\right\}\eps\\[1mm]
\eqref{eq:cvb-strict}&1&
\left\{\frac{151}{50}\eps\right\}\\
\bottomrule
\end{array}
\label{eq:minimal-cvb-slacks}
\end{equation}

\paragraph{Arbitrage.}
Apply \Cref{prop:sfb} with \(a=3\eps\), \(K=5\eps\).  Then
\(\theta=2/3\) and \(\gamma=1/3\).  The pathwise nonnegative strategy is
\[
    \frac23\text{ share}
    -C_1(3\eps)+\frac13C_2(5\eps).
\]
Its setup cost is
\[
\begin{aligned}
\frac23\frac{601}{100}\eps
-\frac{499}{100}\eps
+\frac13\frac{201}{100}\eps
=-\frac{47}{150}\eps.
\end{aligned}
\]
It therefore receives \(47\eps/150\) initially and liquidates
nonnegatively on every path.

\paragraph{Minimality.}
One quoted date cannot yield an augmented counterexample because the source's
one-maturity theorem is a complete characterization.  If all actual calls
occupy a single quoted date, choose a one-maturity consistent reference/shadow
pair for that date and repeat the same laws at every other date, using
identity couplings for the shadow martingale.
If there are no actual calls, choose
\(m\in[S_0^{\bid},S_0^{\ask}]\) and set
\[
    \nu_1=\nu_2=\delta_m,\qquad
    \mu_1=\mu_2=\delta_{\max\{m,\eps\}}.
\]
The reference laws are supported on \([\eps,\infty)\), the shadow laws have
common mean \(m\), and the evident coupling gives
\(W^\infty(\mu_t,\nu_t)\le\eps\).  Thus every panel with fewer than two
actual calls is intertemporally consistent whenever its base conditions
hold.  A genuinely intertemporal obstruction therefore needs at least two
dates and two actual calls.  The panel \eqref{eq:minimal-panel} attains both
lower bounds, with one call at each date.
\end{proof}

\begin{remark}[Scope of minimality]
\label{rem:minimality-scope}
The minimality in \Cref{thm:minimal-panel} concerns the number of quoted
dates and actual call quotes in the finite-call framework of Gerhold and
G\"ul\"um.  It does not assert uniqueness, minimal integer coefficients,
extreme-ray minimality, or a smallest model support.
\end{remark}

\subsection{A full-dimensional open counterexample family}

\begin{corollary}[Open quote box]
\label{cor:open-box}
Fix \(\eps>0\) and the strikes \(3\eps\) and \(5\eps\).  Perturb each of the
six bid and ask coordinates in
\eqref{eq:minimal-stock-quote}--\eqref{eq:minimal-panel} independently by
less than \(\eps/200\).  Every resulting quote panel remains a
base-augmented executable-CVB counterexample with stock-flip arbitrage.
\end{corollary}

\begin{proof}
Every bid--ask width remains between \(\eps/100\) and \(3\eps/100\);
in particular, the initial width remains below \(\eps\).  The call values
\(5\eps\) and \(2\eps\), and the common shadow mean \(6\eps\), remain inside
their respective perturbed quote intervals, so the one-maturity witnesses in
the proof of \Cref{thm:minimal-panel} still apply.

The slack in each non-strict CVB inequality changes by less than
\(\eps/100\): its variable part is one bid coordinate subtracted from one
ask coordinate or a convex combination of two ask coordinates.  The
smallest original slack is \(51\eps/50\), so all non-strict and applicable
strict conditions remain valid.  Finally,
the stock-flip setup cost can increase by at most
\[
    (\theta+1+\gamma)\frac{\eps}{200}
    =\frac{\eps}{100},
\]
whereas its original negative margin is \(47\eps/150\).  The arbitrage
therefore survives throughout the open box.
\end{proof}

\begin{remark}[Status of Conjecture~5.4]
\label{rem:conjecture-status}
Once bid execution is made explicit, the source statement separates into
three claims.  With the CVB bid printed in
\citet{GerholdGulum2020}, even necessity fails by
\Cref{ex:printed-sign}.  With the executable bid
\eqref{eq:executable-cvb-bid}, the corrected inequalities retain their
necessary interpretation, but \Cref{thm:minimal-panel} gives a negative answer
to the sufficiency direction of Conjecture~5.4 even after the base system is
imposed.  This failure holds for every \(\eps>0\) and at minimal date and
actual-call complexity.  The conjecture's separate weak-arbitrage assertion
for panels satisfying the first three CVB families but violating the strict
fourth condition is not addressed here.
\end{remark}

\section{A matching extremal-envelope certificate}
\label{sec:envelope}

The stock-flip strategy proves a pathwise separation.  The same inequality
also follows from the reference/shadow measure geometry, and the two routes
produce the same numerical gap.

\begin{proposition}[Envelope form of the stock-flip bridge]
\label{prop:envelope-sfb}
Under the assumptions of \Cref{prop:sfb}, any \(\eps\)-consistent model must
satisfy \eqref{eq:sfb}.
\end{proposition}

\begin{proof}
Put
\[
    L=K+2\eps.
\]
The coefficients in \eqref{eq:sfb-coefficients} can be written as
\[
    \theta=\frac{L-a}{L-\eps},\qquad
    \gamma=\frac{a-\eps}{L-\eps},\qquad
    \theta+\gamma=1.
\]
For every \(y\ge\eps\), the chord inequality
\begin{equation}
\pos{y-a}
\le\theta\pos{y-\eps}+\gamma\pos{y-L}
\label{eq:chord-inequality}
\end{equation}
holds.  The difference between the right- and left-hand sides is
nonnegative on \([\eps,a]\), decreases linearly to zero on \([a,L]\), and
vanishes for \(y\ge L\).

Suppose a consistent model exists, with reference laws
\(\mu_s,\mu_t\), shadow laws \(\nu_s\le_c\nu_t\), and common shadow mean
\(m\).  Taking \(\mu_s\)-expectations in
\eqref{eq:chord-inequality}, using
\(\supp\mu_s\subseteq[\eps,\infty)\), yields
\begin{equation}
R_{\mu_s}(L)
\ge
\frac{R_{\mu_s}(a)-\theta(E\mu_s-\eps)}{\gamma}
\ge
\frac{\underline r_s(a)-\theta S_0^{\ask}}{\gamma}.
\label{eq:earlier-envelope-lower}
\end{equation}
The second inequality uses
\(|E\mu_s-m|\le\eps\) and \(m\le S_0^{\ask}\).

At the aligned shadow strike
\[
    x=K+\eps=L-\eps,
\]
the lower and upper extremal call functions in
\citet[Proposition~3.2]{GerholdGulum2019} imply
\begin{align}
R_{\nu_s}(x)
&\ge R_{\mu_s}(x+\eps)=R_{\mu_s}(L),
\label{eq:shadow-lower}\\
R_{\nu_t}(x)
&\le R_{\mu_t}(x-\eps)
=R_{\mu_t}(K)
\le\overline r_t(K).
\label{eq:shadow-upper}
\end{align}
Convex order requires \(R_{\nu_s}(x)\le R_{\nu_t}(x)\).
Combining \eqref{eq:earlier-envelope-lower}--\eqref{eq:shadow-upper} gives
\[
    \frac{\underline r_s(a)-\theta S_0^{\ask}}{\gamma}
    \le\overline r_t(K),
\]
which is \eqref{eq:sfb}.
\end{proof}

\begin{corollary}[Matching gap]
\label{cor:matching-gap}
For the panel in \Cref{thm:minimal-panel}, the extremal envelopes force
\[
    R_{\nu_1}(6\eps)\ge\frac{59}{20}\eps,
    \qquad
    R_{\nu_2}(6\eps)\le\frac{201}{100}\eps.
\]
Their gap is \(47\eps/50\).  Multiplying by the terminal call weight
\(\gamma=1/3\) gives \(47\eps/150\), the magnitude of the negative
setup cost.
\end{corollary}

\begin{proof}
Substituting the panel into
\eqref{eq:earlier-envelope-lower} gives
\[
\begin{aligned}
R_{\nu_1}(6\eps)
&\ge
\frac{\frac{499}{100}\eps
-\frac23\frac{601}{100}\eps}{1/3}
=\frac{59}{20}\eps.
\end{aligned}
\]
Equation \eqref{eq:shadow-upper} gives
\[
    R_{\nu_2}(6\eps)\le\frac{201}{100}\eps.
\]
The arithmetic is
\[
    \frac{59}{20}-\frac{201}{100}
    =\frac{47}{50},
\]
and the final statement follows.
\end{proof}

\section{The exact two-date basket operator}
\label{sec:operator}

The preceding separator is one element of a larger executable basket cone.
We now eliminate the arbitrary adapted stock position.

\subsection{Convex trading costs}

The execution rule is a state-dependent convex trading cost, placing the
model within the general transaction-cost superhedging framework.

\begin{lemma}[Encoding as state-dependent convex trading costs]
\label{lem:ckt-encoding}
On the path space \eqref{eq:admissible-state}, changing the stock holding by
\(\Delta\) at state
\((\underline s,x,\overline s)=(x-\beta,x,x+\alpha)\) has cash cost
\begin{equation}
G((\underline s,x,\overline s),\Delta)
=\Delta^+\overline s-\Delta^-\underline s
=\Delta^+(x+\alpha)-\Delta^-(x-\beta)
=\max\{\Delta\underline s,\Delta\overline s\}.
\label{eq:convex-trading-cost}
\end{equation}
This is a state-dependent convex function of \(\Delta\).  Mandatory terminal
liquidation is represented by the final trade \(\Delta=-q\), and the calls
are finite static derivatives of the canonical reference coordinate.
\end{lemma}

\begin{proof}
The bid--ask execution rule gives \eqref{eq:convex-trading-cost} directly.
Its representation as the maximum of two affine functions proves convexity,
and it vanishes at zero.  Liquidating \(q\) shares is the terminal trade
\(-q\).
Finally, a call payoff
\(\pos{x_t-k}\) is a Borel function of the canonical state and its signed
static position can be split into nonnegative long and short quantities
executing at ask and bid.
\end{proof}

On the canonical path space, after splitting each signed static position
into nonnegative long and short claims priced at ask and bid, the model has
the same abstract convex-cost architecture as
\citet[Proposition~3.1]{CheriditoKupperTangpi2017}.  The direct calculation
below exploits the present reference/shadow geometry to obtain an explicit
scalar operator.

\subsection{Static payoffs and the variational operator}

For finite date-\(t\) call positions \(\eta_{t,i}\), set
\begin{equation}
h_t(x)=\sum_{i=1}^{N_t}\eta_{t,i}\pos{x-k_{t,i}},
\qquad t=1,2.
\label{eq:static-payoffs}
\end{equation}
Extend \(h_2\) by \(+\infty\) on \((-\infty,\eps)\).  For an initial stock
holding \(\theta\), define
\begin{equation}
\begin{aligned}
(\cW_{\eps,\theta}h_2)(x)
={}&\theta x+
\sup_{q\in\R}\inf_{y\ge\eps}
\bigl\{
h_2(y)+q(y-x)\\
&\hspace{37mm}
-\eps\bigl(|q-\theta|+|q|\bigr)
\bigr\},
\qquad x\ge\eps.
\end{aligned}
\label{eq:variational-operator}
\end{equation}

\begin{theorem}[Complete two-date executable basket cone]
\label{thm:operator}
Fix \(c,\theta,h_1,h_2\).  The following are equivalent:
\begin{enumerate}[label=\textup{(\roman*)}]
    \item there is a Borel date-one stock rule
    \(q=q(\underline s_1,x_1,\overline s_1)\) such that the associated
    self-financing portfolio has nonnegative final liquidation on every
    admissible two-date path;
    \item there is such a rule depending only on the reference coordinate,
    \(q=q(x_1)\);
    \item
    \begin{equation}
    c+h_1(x)+(\cW_{\eps,\theta}h_2)(x)\ge0
    \qquad\text{for every }x\ge\eps.
    \label{eq:backward-inequality}
    \end{equation}
\end{enumerate}
Consequently, a quote panel admits a two-date model-independent arbitrage if
and only if some finite tuple \((c,\theta,\eta)\) satisfies
\eqref{eq:backward-inequality} and has negative executable cost
\eqref{eq:setup-cost}.
\end{theorem}

\begin{proof}
We give the full reduction because the order of observation and optimization
is essential.

\paragraph{The liquidation identity.}
Write \(x=x_1\), \(y=x_2\), and
\[
\alpha=\overline s_1-x,\qquad
\beta=x-\underline s_1.
\]
Then
\begin{equation}
\alpha,\beta\ge0,\qquad \alpha+\beta\le\eps.
\label{eq:spread-simplex}
\end{equation}
If the investor rebalances from \(\theta\) to \(q\), the date-one bank
holding after the call payoff and stock trade is
\begin{equation}
\begin{aligned}
\phi^0_2(\underline s_1,x,\overline s_1)
={}&c+h_1(x)
-\pos{q-\theta}\,\overline s_1
+\negp{q-\theta}\,\underline s_1.
\end{aligned}
\label{eq:global-bank-account}
\end{equation}
At date two the remaining call and stock positions give
\begin{equation}
\begin{aligned}
L={}&c+h_1(x)+h_2(y)
-\pos{q-\theta}\,\overline s_1
+\negp{q-\theta}\,\underline s_1\\
&+q^+\underline s_2-q^-\overline s_2.
\end{aligned}
\label{eq:exact-liquidation}
\end{equation}
For fixed \(y,q\), the worst terminal spread orientation satisfies
\begin{equation}
\inf_{\underline s_2,\overline s_2}
\bigl(q^+\underline s_2-q^-\overline s_2\bigr)
=qy-\eps|q|.
\label{eq:terminal-spread-minimum}
\end{equation}
At \(y=\eps\), when \(q>0\), the displayed infimum would be attained only at
the excluded zero-bid boundary; it is approached by strictly positive bids.
A negative infimum therefore produces an admissible negative path by
approximation, while a nonnegative infimum bounds all admissible paths.

Substitution of \eqref{eq:spread-simplex} and
\eqref{eq:terminal-spread-minimum} into \eqref{eq:exact-liquidation} gives
the stock contribution
\begin{equation}
\theta x+q(y-x)-\eps|q|
-\alpha\pos{q-\theta}
-\beta\negp{q-\theta}.
\label{eq:oriented-stock-value}
\end{equation}

\paragraph{The observed-spread game.}
For fixed \(x\), let
\begin{equation}
A_x(q)=\inf_{y\ge\eps}\{h_2(y)+q(y-x)\}
\label{eq:Ax}
\end{equation}
and let
\[
\cD_\eps
=\{(\alpha,\beta)\in\R_+^2:\alpha+\beta\le\eps\}.
\]
The investor observes \((\alpha,\beta)\) before choosing \(q\).  Hence the
best guaranteed continuation value is initially
\begin{equation}
\begin{aligned}
\widetilde W(x)
=\theta x+\inf_{(\alpha,\beta)\in\cD_\eps}\sup_{q\in\R}
\bigl\{
A_x(q)-\eps|q|
-\alpha\pos{q-\theta}
-\beta\negp{q-\theta}
\bigr\}.
\end{aligned}
\label{eq:observed-spread-game}
\end{equation}

\paragraph{Uniform compact truncation.}
Let
\[
\Delta_2=\lim_{y\to\infty}\frac{h_2(y)}{y}
=\sum_{i=1}^{N_2}\eta_{2,i},
\qquad q_0=-\Delta_2.
\]
Because \(h_2\) is continuous and piecewise affine,
\begin{equation}
A_x(q)>-\infty
\quad\Longleftrightarrow\quad q\ge q_0.
\label{eq:Ax-domain}
\end{equation}
On \([q_0,\infty)\), \(A_x\) is finite, continuous, concave, and
piecewise affine.

For
\[
\begin{aligned}
F_x(\alpha,\beta,q)
={}&A_x(q)-\eps|q|
-\alpha\pos{q-\theta}
-\beta\negp{q-\theta},
\end{aligned}
\]
define the finite uniform baseline
\begin{equation}
m_x=\min_{(\alpha,\beta)\in\cD_\eps}
F_x(\alpha,\beta,q_0).
\label{eq:uniform-baseline}
\end{equation}
Since
\[
A_x(q)\le h_2(\eps)+q(\eps-x),
\]
we have, for \(q\ge0\) and every spread orientation,
\[
F_x(\alpha,\beta,q)
\le h_2(\eps)+q(\eps-x)-\eps q
\le h_2(\eps)-\eps q.
\]
Thus the explicit choice
\begin{equation}
M_x=
\max\left\{
q_0,\ 0,\
\frac{h_2(\eps)-m_x+1}{\eps}
\right\}
\label{eq:upper-truncation}
\end{equation}
has the property that \(q>M_x\) is strictly worse than the baseline
\(m_x\), uniformly over \(\cD_\eps\).  Every maximizer in
\eqref{eq:observed-spread-game} therefore lies in the compact interval
\([q_0,M_x]\).

On
\(\cD_\eps\times[q_0,M_x]\), \(F_x\) is continuous and affine, hence
convex, in \((\alpha,\beta)\), and continuous and concave in \(q\).  Sion's
minimax theorem gives
\begin{equation}
\inf_{(\alpha,\beta)\in\cD_\eps}\sup_q F_x(\alpha,\beta,q)
=
\sup_q\inf_{(\alpha,\beta)\in\cD_\eps}F_x(\alpha,\beta,q).
\label{eq:sion}
\end{equation}
For fixed \(q\),
\begin{equation}
\inf_{(\alpha,\beta)\in\cD_\eps}
\left\{
-\alpha\pos{q-\theta}
-\beta\negp{q-\theta}
\right\}
=-\eps|q-\theta|.
\label{eq:orientation-minimum}
\end{equation}
Equations \eqref{eq:observed-spread-game}--\eqref{eq:orientation-minimum}
prove
\[
    \widetilde W(x)=(\cW_{\eps,\theta}h_2)(x).
\]

The same truncation also closes the zero-bid boundary mentioned after
\eqref{eq:terminal-spread-minimum}.  On the compact \(q\)-interval the
objective is jointly continuous; Berge's maximum theorem makes the optimized
value continuous in \((\alpha,\beta)\).  The strictly positive admissible
stock states are dense in the closed simplex, so taking its closure changes
neither the infimum nor the pathwise nonnegativity criterion.

\paragraph{Why observing the orientation does not help.}
The supremum on the right of \eqref{eq:sion} is attained.  If \(q(x)\) is a
maximizer, then
\[
F_x(\alpha,\beta,q(x))
\ge \inf_{\cD_\eps}F_x(\cdot,\cdot,q(x))
\]
for every \((\alpha,\beta)\).  Hence one reference-only rule guarantees the
same value as a rule allowed to observe the complete spread orientation.
This proves the equivalence of \textup{(i)} and \textup{(ii)} at the level
of continuation values.

\paragraph{Borel rule and global self-financing strategy.}
\Cref{prop:finite-evaluation} below reduces the supremum to a finite set
\(\cQ(h_2,\theta)\) independent of \(x\).  Each candidate value is a
continuous affine function of \(x\); choosing the smallest maximizing
candidate gives a Borel map \(q(x)\).  Extend it to all positive histories by
\[
(\underline s_1,x,\overline s_1)
\longmapsto q(x\vee\eps),
\]
which is Borel and agrees with the desired selector on the admissible path
space.  With this rule, \eqref{eq:global-bank-account} defines the bank
holding on every positive history and specifies the complete self-financing
account.

Equations \eqref{eq:exact-liquidation}--\eqref{eq:sion} now show that a
pathwise nonnegative continuation exists if and only if
\eqref{eq:backward-inequality} holds.  Combining this with the executable
date-zero cost \eqref{eq:setup-cost} proves all statements.
\end{proof}

\subsection{Finite evaluation}

Let
\[
\mathcal Y=\{\eps\}\cup\{k_{2,i}:1\le i\le N_2\}.
\]

\begin{proposition}[Finite evaluation]
\label{prop:finite-evaluation}
For \(q\ge-\Delta_2\),
\begin{equation}
A_x(q)
=\min_{v\in\mathcal Y}\{h_2(v)+q(v-x)\}.
\label{eq:finite-A}
\end{equation}
The supremum in \eqref{eq:variational-operator} is attained on the finite set
\begin{equation}
\begin{aligned}
\cQ(h_2,\theta)
={}&\{-\Delta_2\}\\
&\cup\bigl(\{0,\theta\}\cap[-\Delta_2,\infty)\bigr)\\
&\cup
\left\{
-\frac{h_2(v)-h_2(u)}{v-u}:
u,v\in\mathcal Y,\ u<v,\
-\frac{h_2(v)-h_2(u)}{v-u}\ge-\Delta_2
\right\}.
\end{aligned}
\label{eq:candidate-set}
\end{equation}
Consequently,
\begin{equation}
\begin{aligned}
(\cW_{\eps,\theta}h_2)(x)
=\theta x+\max_{q\in\cQ(h_2,\theta)}
\Bigl[
&\min_{v\in\mathcal Y}\{h_2(v)+q(v-x)\}\\
&-\eps\bigl(|q-\theta|+|q|\bigr)
\Bigr].
\end{aligned}
\label{eq:finite-operator}
\end{equation}
\end{proposition}

\begin{proof}
The function \(y\mapsto h_2(y)+q(y-x)\) is affine between the terminal
strikes.  When \(q\ge-\Delta_2\), its last affine piece is nondecreasing,
so its infimum is attained at \(\eps\) or at a strike.  This proves
\eqref{eq:finite-A}, including the flat-tail case \(q=-\Delta_2\).

The right side of \eqref{eq:variational-operator} is a concave
piecewise-affine function of \(q\).  Its breakpoints consist of the domain
endpoint \(-\Delta_2\), the penalty kinks \(0,\theta\), and intersections of
the affine pieces in \eqref{eq:finite-A}.  Every breakpoint therefore
belongs to the finite set in \eqref{eq:candidate-set}, although that set may
also contain intersections of inactive affine pieces.  A concave
piecewise-affine function attains its maximum at a breakpoint or on an
interval whose endpoints are breakpoints.
This proves \eqref{eq:finite-operator}.
\end{proof}

For rational inputs, \eqref{eq:finite-operator} is a finite rational
calculation.  It also shows directly that the operator is piecewise affine
as a function of \(x\) and that the smallest-maximizer selector used in
\Cref{thm:operator} is Borel.

\subsection{Convex-envelope representation and attainment}

\begin{proposition}[Closed-form convex-envelope operator]
\label{prop:convex-envelope}
Let \(g=h_2^{**}\) be the closed convex envelope of the extended terminal
payoff on \([\eps,\infty)\).  Then \(g\) is proper and, for every
\(x\ge\eps\),
\begin{equation}
\boxed{
(\cW_{\eps,\theta}h_2)(x)
=\theta x-\eps|\theta|
+\inf_{\substack{z\ge\eps\\|z-x|\le2\eps}}
\left\{
g(z)+\pos{\theta(z-x)}
\right\}.
}
\label{eq:convex-envelope-operator}
\end{equation}
The infimum is attained.
\end{proposition}

\begin{proof}
The finite piecewise-affine function \(h_2\) has an affine minorant on
\([\eps,\infty)\), so its biconjugate \(g\) is a proper closed convex
function.  Conjugation does not distinguish a function from its closed
convex envelope, and
\[
    A_x(q)=-g^*(-q)-qx.
\]
Setting \(p=-q\) in \eqref{eq:variational-operator} gives
\begin{equation}
\cW_{\eps,\theta}h_2(x)
=\theta x+(g^*+\varphi_\theta)^*(x),
\qquad
\varphi_\theta(p)=\eps\bigl(|p|+|p+\theta|\bigr).
\label{eq:conjugate-sum}
\end{equation}
A direct calculation yields
\begin{equation}
\varphi_\theta^*(r)
=
\begin{cases}
\displaystyle
\max(0,-\theta r)-\eps|\theta|,
&|r|\le2\eps,\\[1mm]
+\infty,&|r|>2\eps.
\end{cases}
\label{eq:penalty-conjugate}
\end{equation}

Because \(\varphi_\theta\) is finite and continuous on \(\R\), the
sum-conjugate qualification is automatic: the infimal convolution has no
closure gap.  Substituting \eqref{eq:penalty-conjugate} into
\eqref{eq:conjugate-sum}, with \(r=x-z\), yields
\eqref{eq:convex-envelope-operator}.  The remaining domain
\[
[\max\{\eps,x-2\eps\},\,x+2\eps]
\]
is nonempty and compact.  The objective is lower semicontinuous there, so
the infimum is attained.
\end{proof}

For zero initial stock, the formula reduces to the especially simple
operation
\begin{equation}
\cW_{\eps,0}h_2(x)
=\inf_{\substack{z\ge\eps\\|z-x|\le2\eps}}h_2^{**}(z):
\label{eq:zero-stock-operator}
\end{equation}
Thus one first takes the lower convex envelope and then minimizes it over the
\(2\eps\)-neighborhood produced by the two worst spread crossings.

\subsection{Recovery of the stock flip}

\begin{corollary}[Operator explanation of the bridge]
\label{cor:operator-sfb}
For the coefficients in \eqref{eq:sfb-coefficients}, take
\[
    h_2(y)=\gamma\pos{y-K},
    \qquad h_1(x)=-\pos{x-a}.
\]
Then
\begin{equation}
\cW_{\eps,\theta}h_2(x)
=\max\{\theta(x-\eps),\,x-a\},
\qquad x\ge\eps,
\label{eq:sfb-operator-value}
\end{equation}
and \(h_1+\cW_{\eps,\theta}h_2\ge0\).
\end{corollary}

\begin{proof}
The candidate set contains \(q\in\{-\gamma,0,\theta\}\).  Their
operator values are
\[
\begin{array}{c@{\qquad}c}
\toprule
q&\text{value}\\
\midrule
-\gamma&x-a\\
0&\theta(x-\eps)\\
\theta&0\\
\bottomrule
\end{array}
\]
The first identity uses
\(\theta+\gamma=1\) and
\(\gamma(K+\eps)=a-\eps\).  There are no larger values because
\eqref{eq:finite-operator} exhausts the concave breakpoints.  Since
\(\theta(x-\eps)\ge0\), the maximum of the three entries is
\eqref{eq:sfb-operator-value}.  Subtracting the earlier call payoff leaves a
nonnegative function.
\end{proof}

The two branches are generated by \(q=0\) and \(q=-\gamma\).  The
\(q=0\) branch is maximizing for \(x\le K+2\eps\), and the
\(q=-\gamma\) branch is maximizing for \(x\ge K+2\eps\).  The elementary
hedge in \Cref{prop:sfb} switches earlier, at \(x=a\); it remains
nonnegative but is not continuation-optimal on \([a,K+2\eps)\).

\section{Pairwise bridges and the full basket cone}
\label{sec:limits}

The minimal counterexample identifies one missing pairwise basket.  It is
natural to ask whether adding every stock-flip bridge to the CVB system would
restore sufficiency.  It does not.

\begin{proposition}[Pairwise stock-flip bridges do not close the cone]
\label{prop:pairwise-insufficient}
At \(\eps=1\), consider
\[
[S_0^{\bid},S_0^{\ask}]
=\left[\frac{499}{100},\frac{501}{100}\right]
\]
and the four actual call quotes
\begin{equation}
\renewcommand{\arraystretch}{1.55}
\begin{array}{c@{\quad}c@{\quad}c@{\quad}c}
\toprule
t&k&\underline r&\overline r\\
\midrule
1&\frac32&\frac{269}{100}&\frac{27}{10}\\[1mm]
1&\frac52&\frac{21}{10}&\frac{211}{100}\\[1mm]
2&\frac52&\frac{17}{5}&\frac{341}{100}\\[1mm]
2&\frac92&\frac{149}{100}&\frac32\\
\bottomrule
\end{array}
\label{eq:four-call-panel}
\end{equation}
The panel satisfies the base system.  Every corrected executable CVB
condition and every stock-flip bridge formed from one actual call at each
date holds strictly.  Nevertheless it admits a model-independent arbitrage
with setup cost \(-9/40\).
\end{proposition}

\begin{proof}
Separate one-maturity witnesses and the CVB calculations are given in
\Cref{app:calculations}.

\begin{samepage}
The four pairwise stock-flip slacks are
\[
\renewcommand{\arraystretch}{1.25}
\begin{array}{c@{\quad}c@{\quad}c}
\toprule
a&K&\theta S_0^{\ask}-\underline r_1(a)
       +\gamma\overline r_2(K)\\
\midrule
\frac32&\frac52&\frac{366}{175}\\[1mm]
\frac32&\frac92&\frac{2201}{1100}\\[1mm]
\frac52&\frac52&\frac{1557}{700}\\[1mm]
\frac52&\frac92&\frac{537}{275}\\
\bottomrule
\end{array}
\]
so every pairwise bridge is strict.
\end{samepage}

Now take \(c=5/2\), \(\theta=0\), and
\begin{align*}
h_1(x)
&=\pos{x-\frac32}-\pos{x-\frac52},\\
h_2(y)
&=-\frac74\pos{y-\frac52}
+\frac74\pos{y-\frac92}.
\end{align*}
Write
\[
H(y)=\frac74\left[
2-\pos{y-\frac52}+\pos{y-\frac92}
\right],
\qquad
h_2=H-\frac72.
\]
The closed convex envelope of \(H\) on \([1,\infty)\) is
\(\pos{\frac92-y}\).  Indeed, this function lies below \(H\): on
\([1,5/2]\) it is bounded by \(7/2=H\), on \([5/2,9/2]\) it is
\(9/2-y\le(7/4)(9/2-y)=H(y)\), and both functions vanish thereafter.
Conversely, if \(f\) is convex and \(f\le H\), then
\(f(1)\le7/2\) and \(f(9/2)\le0\).  Convexity bounds \(f\) on
\([1,9/2]\) by the chord \(9/2-y\), while \(f\le0\) for
\(y\ge9/2\).  Therefore
\Cref{prop:convex-envelope}, with zero initial stock, gives
\[
\cW_{1,0}h_2(x)
=\pos{\frac52-x}-\frac72.
\]
For every \(x\ge1\),
\[
\frac52+h_1(x)+\cW_{1,0}h_2(x)
=
\begin{cases}
\frac32-x,&1\le x\le\frac32,\\
0,&x\ge\frac32.
\end{cases}
\]
The portfolio is pathwise nonnegative.  Its executable setup cost is
\[
\begin{aligned}
\frac52+\frac{27}{10}-\frac{21}{10}
-\frac74\!\cdot\!\frac{17}{5}
+\frac74\!\cdot\!\frac32
=-\frac9{40}.
\end{aligned}
\]
Thus no collection consisting only of the corrected CVBs and all available
pairwise stock-flip bridges can describe the full basket cone.
\end{proof}

\subsection{Model-independent arbitrage versus consistency}

\Cref{thm:operator} is a complete characterization of the
\emph{model-independent-arbitrage cone}.  It is not by itself a
characterization of the probabilistic consistency set
\eqref{eq:measure-consistency}.  Consistency always rules out
model-independent arbitrage, but the converse can fail on weak-arbitrage
boundaries.  This distinction already occurs in the source's one-maturity
theorem when its equality condition fails: under every candidate model, some
model-dependent arbitrage exists, yet there is no single pathwise strategy
that is strict on all possible models.

The strict counterexample in \Cref{thm:minimal-panel} lies away from this
boundary and is unaffected by the distinction.  A complete identification
of the operator cone with consistency on the remaining equality faces would
require an additional closure or density theorem.

\subsection{Scope}

The counterexample concerns the corrected, base-augmented executable CVB
system; the distinct weak-arbitrage clause of Conjecture~5.4 remains open.
The minimality statement counts quoted dates and actual call quotes in the
finite-call framework of Gerhold and G\"ul\"um.  The operator is complete
for two-date finite-call baskets as the static coefficients range over
\(\R\), but it does not provide a finite irredundant representation of that
cone.  Extensions to three or more dates and the frictionless limit
\(\eps=0\) are beyond the present scope.

\subsection{Conclusion}

Correcting the CVB bid does not repair the multi-maturity sufficiency
problem.  The corrected executable system fails for every positive spread
bound at the first possible intertemporal quote complexity: two dates and one
actual call at each.  The stock flip gives both a transparent executable
arbitrage and a matching convex-order separation, and the failure is stable
on an open set of quotes.

The backward operator shows why no small list of single transported
calls should be expected to capture the whole pathwise geometry.  It first
convexifies the terminal static payoff and then searches a
\(2\eps\)-neighborhood generated by the two adverse spread crossings, with
an additional one-sided term from the initial stock.  The minimal stock flip
is the smallest visible instance of that broader basket geometry.

\section{Use of generative AI}
\label{sec:aiuse}

This work used two generative AI configurations: \mbox{OpenAI} GPT-5.6 Sol
Pro through ChatGPT and \mbox{OpenAI} GPT-5.6 Sol Ultra through Codex.  The
ChatGPT configuration was used in the initial analysis of the source
formulation and in exploratory counterexample searches.  The Codex
configuration was used in developing the minimal two-call construction, its
open family, and the two-date basket operator.  Both systems also assisted
with proofs and manuscript preparation.

The author chose the research problem, directed the work, and decided which
arguments and formulations to include.

\appendix
\section{Supporting finite calculations}
\label{app:calculations}

\subsection{Base inequalities for the minimal panel}

Besides the explicit one-maturity witnesses in
\Cref{thm:minimal-panel}, the source's finite one-maturity inequalities have
the following positive slacks when the auxiliary quote is included:
\[
\renewcommand{\arraystretch}{1.40}
\begin{array}{c@{\qquad}c@{\qquad}c}
\toprule
t&\text{slope lower bound}&\text{monotonicity}\\
\midrule
1&\frac{151}{50}\eps&\frac{51}{50}\eps\\[1mm]
2&\frac{101}{50}\eps&\frac{201}{50}\eps\\
\bottomrule
\end{array}
\]
There is only one actual call at each date, hence no butterfly condition.
Since both monotonicity slacks are strictly positive, the zero-cost
call-spread case does not arise.

\subsection{Four-call panel: one-maturity witnesses}

For the panel \eqref{eq:four-call-panel}, let
\[
\mu_1=\frac25\delta_1+\frac35\delta_6.
\]
Then
\[
R_{\mu_1}\!\left(\frac32\right)=\frac{27}{10},
\qquad
R_{\mu_1}\!\left(\frac52\right)=\frac{21}{10},
\qquad E\mu_1=4.
\]
Translate \(\mu_1\) one unit to the right to obtain a shadow law
\(\nu_1\) of mean \(5\) at \(W^\infty\)-distance \(1\).

At date two, let
\[
\mu_2=\frac15\delta_4+\frac45\delta_{51/8}.
\]
Then
\[
R_{\mu_2}\!\left(\frac52\right)=\frac{17}{5},
\qquad
R_{\mu_2}\!\left(\frac92\right)=\frac32,
\qquad E\mu_2=\frac{59}{10}.
\]
Translate \(\mu_2\) by \(-9/10\) to obtain a shadow law \(\nu_2\) of mean
\(5\) at \(W^\infty\)-distance \(9/10\).  Both dates are therefore
separately consistent with the initial interval
\([499/100,501/100]\).

\subsection{Four-call panel: corrected CVB inequalities}

The auxiliary quote is
\[
k_{t,0}=1,\qquad
\underline r_{t,0}=\frac{299}{100},\qquad
\overline r_{t,0}=\frac{501}{100}.
\]
The six CVB points are
\[
\renewcommand{\arraystretch}{1.55}
\begin{array}{c@{\quad}c@{\quad}c@{\quad}c}
\toprule
j&\sigma&z&b^{\mathrm{exec}}\\
\midrule
0& 1&0&\frac{299}{100}\\
0&-1&2&\frac{99}{100}\\
1& 1&\frac12&\frac{269}{100}\\
1&-1&\frac52&\frac{69}{100}\\
2& 1&\frac32&\frac{21}{10}\\
2&-1&\frac72&\frac1{10}\\
\bottomrule
\end{array}
\]
and the shifted terminal ask nodes are
\[
\left(2,\frac{501}{100}\right),\qquad
\left(\frac72,\frac{341}{100}\right),\qquad
\left(\frac{11}{2},\frac32\right).
\]
Define
\[
U(x)=
\begin{cases}
\dfrac{69}{10}-x,
&x\le\dfrac72,\\[2mm]
\dfrac{17}{5}-\dfrac{19}{20}
\left(x-\dfrac72\right),
&\dfrac72\le x\le\dfrac{269}{38},\\[2mm]
0,&x\ge\dfrac{269}{38}.
\end{cases}
\]
Its slopes are \(-1,-19/20,0\).  Hence \(U\) is nonnegative,
nonincreasing, and convex with slope bounded below by \(-1\).  At the three
terminal nodes, the ask-minus-\(U\) gaps are
\[
    \frac{11}{100},\qquad \frac1{100},\qquad 0,
\]
while the \(U(z)-b^{\mathrm{exec}}\) gaps, in the row order above, are
\[
\frac{391}{100},\quad
\frac{391}{100},\quad
\frac{371}{100},\quad
\frac{371}{100},\quad
\frac{33}{10},\quad
\frac{33}{10}.
\]
\Cref{lem:call-geometry} proves all corrected CVB faces.

The numbers of instances of
\eqref{eq:cvb-chord}, \eqref{eq:cvb-slope},
\eqref{eq:cvb-monotone}, and \eqref{eq:cvb-strict} are, respectively,
\[
    (3,14,4,2).
\]
The positive slacks on the three chord faces and the two applicable strict
monotonicity faces are, respectively,
\[
\left(\frac{284}{75},\frac{2673}{700},\frac{596}{175}\right)
\qquad\text{and}\qquad
\left(\frac{108}{25},\frac{491}{100}\right).
\]
All remaining finite slacks are positive as certified by the displayed
function \(U\).

\begingroup
\small
\phantomsection

\endgroup

\end{document}